\documentclass[letterpaper,USenglish,numberwithinsect,cleveref]{lipics-v2019}

\usepackage{xcolor}
\definecolor{darkgray}{gray}{0.4}
\definecolor{lightgray}{gray}{0.75}
\definecolor{midgray}{gray}{0.6}
\usepackage{tikz}
\usetikzlibrary{arrows,shapes,trees, arrows.meta}

\graphicspath{{./figures/}}

\newcommand{\twoclass}[1]{\ensuremath{(#1)\text{-}\mathsf{Color}}}
\newcommand{\Twoclass}[1]{Two-Class \ensuremath{(#1)}-Coloring}
\newcommand{\kColor}[0]{\ensuremath{k\text{-}\mathsf{Color}}}
\newcommand{\maxIS}[0]{\ensuremath{\mathsf{MaxIS}}}
\newcommand{\maxCut}[0]{\ensuremath{\mathsf{MaxCut}}}
\newcommand{\maxKCut}[0]{\ensuremath{\mathsf{Max}\text{-}k\text{-}\mathsf{Cut}}}

\makeatletter
\@fleqnfalse
\makeatother

\hideLIPIcs

\title{\Twoclass{r,k}:\\ Coloring with Service Guarantees\footnote{The authors of this paper are alphabetically ordered.}} 

\titlerunning{\Twoclass{r,k}}

\author{Pál András Papp}{ETH Zurich, Switzerland}{apapp@ethz.ch}{}{}

\author{Roland Schmid}{ETH Zurich, Switzerland}{roschmi@ethz.ch}{}{}

\author{Valentin Stoppiello}{ETH Zurich, Switzerland}{stoppiev@ethz.ch}{}{}

\author{Roger Wattenhofer}{ETH Zurich, Switzerland}{wattenhofer@ethz.ch}{}{}

\authorrunning{P.\,A.\ Papp and R.\ Schmid and V.\ Stoppiello and R.\ Wattenhofer}

\Copyright{P.\ Andras Papp and Roland Schmid and Valentin Stoppiello and Roger Wattenhofer}



\EventEditors{John Q. Open and Joan R. Access}
\EventNoEds{2}
\EventLongTitle{42nd Conference on Very Important Topics (CVIT 2016)}
\EventShortTitle{CVIT 2016}
\EventAcronym{CVIT}
\EventYear{2016}
\EventDate{December 24--27, 2016}
\EventLocation{Little Whinging, United Kingdom}
\EventLogo{}
\SeriesVolume{42}
\ArticleNo{23}

\begin{document}

\maketitle

\begin{abstract}
This paper introduces the \emph{\Twoclass{r,k}} problem: Given a fixed number of $k$ colors, such that only $r$ of these $k$ colors allow conflicts, what is the minimal number of conflicts incurred by an optimal coloring of the graph?

We establish that the family of \Twoclass{r,k} problems is NP-complete for any $k \geq 2$ when $(r, k) \neq (0,2)$. Furthermore, we show that \Twoclass{r,k} for $k \geq 2$ colors with one ($r = 1$) relaxed color cannot be approximated to any constant factor ($\notin$ APX). Finally, we show that \Twoclass{r,k} with $k \geq r \geq 2$ colors is APX-complete.
\end{abstract}


\section{Introduction}
\label{sec:introduction}

Graph coloring is a fundamental mathematical problem with various applications in computer science. Usually, we want to color the nodes of a graph with a minimum number of colors such that no two adjacent nodes have the same color. 

In many real world applications, however, the number of available colors is strictly limited. For example, we are given a fixed set of frequencies, and we need to assign a frequency (color) to each wireless transmitter (node in a graph), such that wireless transmitters that interfere with each other (that are connected by an edge in the graph) do not use the same frequency. 
In cases where the set of colors is fixed, we may need to accept a certain number of conflicts, i.e., neighboring nodes which have the same color.

In practice, we often have coloring problems that additionally require strict service guarantees, e.g.\ priority service for premium customers, or unobstructed communication channels for maintenance and emergency personnel.
Such applications can be modeled by partitioning the set of colors into (i) \emph{proper} colors for guaranteed conflict-free service and (ii) \emph{relaxed} colors which allow conflicts. In general, we can have $k$ colors, with $r$ of these $k$ colors being relaxed and $k-r$ of them being proper; we refer to this setting as the \Twoclass{r,k} problem.

\paragraph*{Contributions}
In this paper, we introduce and motivate the \emph{\Twoclass{r,k}} problem, that is, the problem of coloring a given graph with a fixed number of $k$ colors of which $r$ colors are relaxed to allow conflicts.
We establish relations between the \Twoclass{r,k} problem and several well-studied problems, such as Maximum Independent Set and \maxKCut{}.
For our main contributions, we show that:
\begin{itemize}
    \item the family of \Twoclass{r,k} problems is NP-complete, except for the special case when $(r,k) = (0,2)$,
    \item \Twoclass{1,k}, i.e., with one relaxed color, cannot be approximated to any constant factor within polynomial time ($\notin$ APX),
    \item \Twoclass{r,k} with multiple relaxed colors is APX-complete, i.e., \Twoclass{r,k} can be approximated to some constant (but not to an arbitrary constant) in polynomial time.
\end{itemize}

\section{Related Work}

Graph coloring is one of the most fundamental and well-studied problems of computer science. A general survey of fundamental graph coloring results is available in \cite{coloringBook1} or \cite{coloringBook2}. Even though it has been studied intensively since the early 1970s, graph coloring still receives significant attention today, with results ranging from distributed algorithms \cite{furtherAreas1} to modified variants \cite{coloringBook1} or heuristic solutions for such problems \cite{furtherAreas2}. However, the vast majority of these studies only consider proper colorings of graphs, i.e.\ when monochromatic edges are not allowed at all.

The most well-studied relaxed coloring problem is \emph{Defective Coloring}, introduced in \cite{DefectiveIntro}. A defective coloring is a (possibly relaxed) coloring of a graph $G$ such that each node of $G$ has at most $d$ conflicts. Given a graph $G$ and a fixed number of colors $k$, the goal of the problem is to find a $k$-coloring with the minimal possible $d$ value. Defective coloring is also known to be NP-complete \cite{DefectiveHard}, and has been studied extensively on specific classes of graphs \cite{Defective1, Defective2, Defective3}.

In the case when all color classes allow conflicts, a solution minimizing the number of conflicts is equivalently a solution maximizing the number of edges that go between different color classes. Thus the \Twoclass{k,k} problem is essentially a reformulation of the \maxKCut{} problem. The \maxKCut{} problem is also a thoroughly studied problem, which is known to be both NP-complete and APX-complete \cite{NPCompleteProblemsBook,ApproxHard, APXcomplete}. Further theoretical work mostly investigates the best possible constant-factor approximation for the problem \cite{MaxkCut1, MaxkCut2, MaxkCut3}. Many of these studies focus primarily on the weighted version of the problem, and obtain their results for the unweighted graphs as a special case of this.

Another related question is the Maximum $k$-Colorable Subgraph problem: given a graph, the task here is to select the largest subset $V_0$ of the vertices such that the subgraph induced by $V_0$ still admits a proper $k$-coloring \cite{ColorSubgraph1, ColorSubgraph2}. Note that the term ``Maximum $k$-Colorable Subgraph problem'' is quite ambiguous in the literature, as it is often also used to refer to the problem of finding a $k$-colorable (not necessarily induced) subgraph with maximally many edges, which is essentially only another reformulation of \maxKCut{}.

Our work is also closely connected to the Maximum Independent Set (\maxIS{}) problem, which is known to be NP-complete \cite{NPCompleteProblemsBook}, and also not approximable to any constant \cite{ApproxHard}.


\section{\Twoclass{r,k}} 
We have seen that several relaxations of the traditional graph coloring problem have been considered in the literature. In this work, we introduce the problem of \emph{\Twoclass{r,k}}, which follows a utilitarian approach in the sense that we minimize the global number of conflicts.
To begin with, let us introduce some notation.

As usual, a graph~$G = (V, E)$ consists of a set of $n$ vertices~$V$ (with $|V| = n$) and a set of undirected edges~$E \subseteq {V \choose 2}$. The degree of a node $v \in V$ is denoted by $d(v)$.

For such a graph, a \emph{$k$-coloring} is a mapping of its $n$ vertices to a set of $k$ colors.
Traditionally, any set of nodes that is assigned the same color is required to form an independent set in the graph. 

\begin{definition}
    The assignment of the same color to two adjacent vertices is called a \emph{conflict}.
    Similarly, the edge connecting these two vertices is called a \emph{conflict edge} (for being a witness of the conflict). All other edges are called \emph{covered edges}.
\end{definition}
Graph Coloring does not allow conflicts at all. In our problems, however, the $k$ colors are sorted into two groups, and we only require the colors in one of the groups to form independent sets.
\begin{definition}
     A color is said to be a \emph{proper color} if the vertices mapped to this color are required to form an independent set (i.e., have no conflicts). Otherwise, the color is called a \emph{relaxed color}, which might have (or might not have) conflicts.
\end{definition}
Given these definitions, we can now define the colorings that we are interested in:
\begin{definition}
    In the \emph{\Twoclass{r,k}} problem (with $r \leq k$), a feasible solution is a $k$-coloring of the input graph such that we can divide the $k$ colors into a group of $r$~relaxed colors and a group of $k-r$ proper colors.
\end{definition}
The natural goal in such a problem is to color the graph with as few conflicts as possible. However, when discussing the connection to related problems, it is beneficial to look at this problem in the dual perspective of having as many covered edges as possible. Note that maximizing the number of covered edges in a coloring indeed produces a solution that also minimizes the number of conflicts. Thus, we study the following problem in the paper.

\begin{definition} The \emph{\Twoclass{r,k}} problem is defined as follows:
\begin{itemize}
\item \textbf{Decision problem:} Given a graph $G$ and numbers $r, k$, and $c$, is there a feasible solution of \twoclass{r,k} which covers at least $c$ edges?
\item \textbf{Optimization problem:} Given a graph $G$ and numbers $r$ and $k$, what is the maximal number of edges in $G$ that can be covered by a feasible solution of \twoclass{r,k}?
\end{itemize}
\end{definition}

\noindent Let us finish with some further remarks on notations. Generally, given an optimization problem $P$, we denote the value of the optimal solution on graph $G$ by $\texttt{opt}_{P}(G)$. For an arbitrary feasible solution $\mathcal{S}$ of the problem, the value of $\mathcal{S}$ is denoted by $\texttt{val}_{P}(\mathcal{S})$.

While we study all the above problems on simple graphs only, we also use multigraphs as a tool in one of our proofs. A multigraph is a graph that allows parallel edges between nodes, i.e.\ where $E$ is a multiset of elements from ${V \choose 2}$.

\paragraph*{Example: \Twoclass{r,k} is Not Equivalent to \maxIS{}}
Let us briefly compare \twoclass{1,2} and the Maximum Independent Set (\maxIS{}) problem on the example graph of \Cref{fig:exemplify-difference-MaxIS}. Both problems can be formulated as a coloring problem with one relaxed (white) and one proper (black) color, so the set of feasible solutions for the two problems are identical. However, while \maxIS{} aims to maximize the number of black vertices, \twoclass{1,2} maximizes the number of edges covered by the black vertices instead. \Cref{fig:exemplify-difference-MaxIS} shows that these problems may have different optimal solutions: the optimum of \twoclass{1,2} has only 5 black nodes and covers 13 edges (left), while the Maximum Independent Set contains 10 nodes and covers 11 edges (right).
\begin{figure}[tbh]
    \centering
    \begin{subfigure}[c]{0.51\textwidth}
        \centering
        \includegraphics[height=2cm]{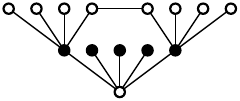}
        \subcaption{Optimal \twoclass{1,2} solution covering 13 edges.}
    \end{subfigure}
    \hfill
    \begin{subfigure}[c]{0.47\textwidth}
        \centering
        \includegraphics[height=2cm]{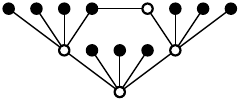}
        \subcaption{Maximum independent set of size $10$.}
    \end{subfigure}
    \caption{The \twoclass{1,2} problem is not equivalent to the Maximum Independent Set problem.}
    \label{fig:exemplify-difference-MaxIS}
\end{figure}

\section{Complexity Results (NP-Completeness)} 
\label{sec:NP-completeness}

In this section, we show that  \twoclass{r,k} is NP-complete for $k > 2$ with $0 \leq r \leq k$, and also for $k = 2$ with $r \in \{1,2\}$. To begin with, observe that deciding \twoclass{r,k} is in NP: A solution to \twoclass{r,k} covering at least $c$ edges can be verified in polynomial time by traversing the set of covered edges.

Now, for the case $k > 2$, recall that deciding the traditional graph coloring problem \kColor{} is NP-complete for $k > 2$ colors:
\begin{theorem}[\cite{NPCompleteProblemsBook}]
    It is NP-complete to decide whether an input graph $G$ admits a $k$-coloring for a given number of proper colors $k > 2$.
\end{theorem}
To show that \twoclass{r,k} is NP-complete for $k > 2$, we may thus show that a polynomial-time algorithm for solving \twoclass{r,k} could be used to compute a solution for 
\kColor{} in polynomial time. To that end, note that setting $c = |E|$ directly yields a decision procedure for \kColor{}.
\begin{lemma}
    \label{lem:k-larger-2}
    For $k > 2$, it is NP-complete to decide whether an input graph $G$ admits a solution to \twoclass{r,k} covering at least $c$ edges.
\end{lemma}

The case for $k = 2$ must be treated separately, as \kColor{} for $k = 2$ colors (deciding bipartiteness of a graph) can be solved in linear time, e.g.\ using a breadth-first search. Hence, \twoclass{0,2} can be decided in linear time for $c = |E|$. For $c < |E|$, note that if there exists a proper 2-coloring of the input graph $G$, then it covers all $|E|$ edges and is hence a solution to \twoclass{0,2} for any $c < |E|$. If there exists no proper 2-coloring of $G$ but a solution $\mathcal{S}$ to \twoclass{0,2} covering at least $c < |E|$ edges, $\mathcal{S}$ would have to induce at least one conflict. However, in \twoclass{0,2} there are no relaxed colors that would allow for any conflicts. Consequently, such a solution $\mathcal{S}$ may not exist and \twoclass{0,2} can be decided for any $c$ in linear time.

Perhaps surprisingly, we will show that \twoclass{r,k} is also NP-complete for $k = 2$ and $r > 0$. In other words, relaxing at least one color creates a gap in the computational complexity (assuming P $\neq$ NP). 

To show that \twoclass{r,2} is also NP-hard to decide for $r > 0$ and some given minimum number of covered edges $c$, note that \twoclass{2,2} is equivalent to the well-known \maxCut{} problem which is known to be NP-complete \cite{NPCompleteProblemsBook}. For \twoclass{1,2}, we propose a reduction to the Maximum Independent Set problem. To begin with, let us recall that:
\begin{theorem}[\cite{NPCompleteProblemsBook}]
    It is NP-complete to decide whether an input graph $G$ contains an independent set of some given size $k$.
\end{theorem}
We propose the following reduction: Given an input graph $G = (V_G, E_G)$, we determine whether there exists an independent set of size $k$ as follows:
\begin{enumerate}
    \item add a clique $C$ of size $|V_C| = n^2$ to the graph $G$ and connect each clique-vertex $u \in V_C$ with each graph vertex $v \in V_G$ to obtain the modified graph $G' = (V_{G'}, E_{G'})$, as shown in Figure \ref{fig:addclique}
    \item decide whether there is a solution for \twoclass{1,2} on $G'$ covering at least $c = k \cdot n^2$ edges,
    \item output the same answer for whether an independent set of size $k$ exists in $G$.
\end{enumerate}

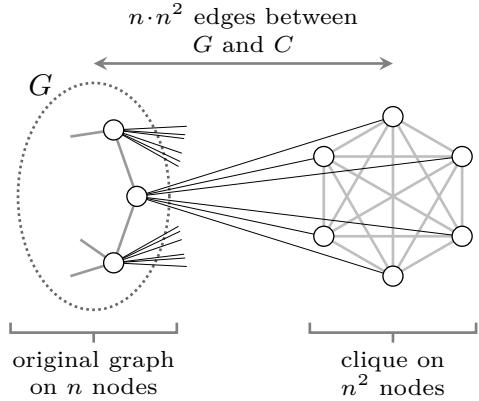
\begin{figure}
	\captionsetup{justification=centering}
	\centering
	\resizebox{!}{5.5cm}{%
	    \begin{tikzpicture}
	
	\draw[darkgray, thick, densely dotted] (30pt,50pt) ellipse (0.8cm and 1.2cm);
	
	\draw[midgray, thick] (43pt,50pt) -- (36pt,70pt);
	\draw[midgray, thick] (43pt,50pt) -- (36pt,30pt);
	\draw[midgray, thick] (36pt,30pt) -- (23pt,26pt);
	\draw[midgray, thick] (36pt,30pt) -- (26pt,37pt);
	\draw[midgray, thick] (36pt,70pt) -- (23pt,68pt);
	
	\draw[lightgray, thick] (120pt,74pt) -- (99.2pt,62pt);
	\draw[lightgray, thick] (120pt,74pt) -- (99.2pt,38pt);
	\draw[lightgray, thick] (120pt,74pt) -- (140.8pt,62pt);
	\draw[lightgray, thick] (120pt,74pt) -- (140.8pt,38pt);
	\draw[lightgray, thick] (120pt,74pt) -- (120pt,26pt);
	\draw[lightgray, thick] (99.2pt,62pt) -- (99.2pt,38pt);
	\draw[lightgray, thick] (99.2pt,62pt) -- (140.8pt,62pt);
	\draw[lightgray, thick] (99.2pt,62pt) -- (140.8pt,38pt);
	\draw[lightgray, thick] (99.2pt,62pt) -- (120pt,26pt);
	\draw[lightgray, thick] (99.2pt,38pt) -- (140.8pt,62pt);
	\draw[lightgray, thick] (99.2pt,38pt) -- (140.8pt,38pt);
	\draw[lightgray, thick] (99.2pt,38pt) -- (120pt,26pt);
	\draw[lightgray, thick] (140.8pt,62pt) -- (140.8pt,38pt);
	\draw[lightgray, thick] (140.8pt,62pt) -- (120pt,26pt);
	\draw[lightgray, thick] (140.8pt,38pt) -- (120pt,26pt);
	
	\draw[very thin] (43pt,50pt) -- (120pt,74pt);
	\draw[very thin] (43pt,50pt) -- (99.2pt,62pt);
	\draw[very thin] (43pt,50pt) -- (99.2pt,38pt);
	\draw[very thin] (43pt,50pt) -- (140.8pt,62pt);
	\draw[very thin] (43pt,50pt) -- (140.8pt,38pt);
	\draw[very thin] (43pt,50pt) -- (120pt,26pt);
	
	\draw[very thin] (36pt,70pt) -- (58pt,71.1pt);
	\draw[very thin] (36pt,70pt) -- (57.2pt,67.4pt);
	\draw[very thin] (36pt,70pt) -- (56pt,60.5pt);
	\draw[very thin] (36pt,70pt) -- (57.5pt,68.55pt);
	\draw[very thin] (36pt,70pt) -- (56.5pt,63pt);
	\draw[very thin] (36pt,70pt) -- (55.5pt,59pt);	
	
	\draw[very thin] (36pt,30pt) -- (58pt,28.9pt);
	\draw[very thin] (36pt,30pt) -- (57.2pt,32.6pt);
	\draw[very thin] (36pt,30pt) -- (56pt,39.5pt);
	\draw[very thin] (36pt,30pt) -- (57.5pt,31.45pt);
	\draw[very thin] (36pt,30pt) -- (56.5pt,37pt);
	\draw[very thin] (36pt,30pt) -- (55.5pt,41pt);	
	
	\draw[black, fill=white] (43pt,50pt) circle (0.7ex);
	\draw[black, fill=white] (36pt,70pt) circle (0.7ex);
	\draw[black, fill=white] (36pt,30pt) circle (0.7ex);
	
	\draw[black, fill=white] (120pt,74pt) circle (0.7ex);
	\draw[black, fill=white] (99.2pt,62pt) circle (0.7ex);
	\draw[black, fill=white] (99.2pt,38pt) circle (0.7ex);
	\draw[black, fill=white] (140.8pt,62pt) circle (0.7ex);
	\draw[black, fill=white] (140.8pt,38pt) circle (0.7ex);
	\draw[black, fill=white] (120pt,26pt) circle (0.7ex);
	
	
	\draw[gray, thick] (5pt,12pt) -- (5pt,9pt) -- (55pt,9pt) -- (55pt,12pt);
	\draw[gray, thick] (30pt,6pt) -- (30pt,9pt);
	\node[anchor=north] at (30pt,7pt) {\scriptsize original graph};
	\node[anchor=north] at (30pt,-1pt) {\scriptsize on $n$ nodes};

	\draw[gray, thick, arrows=stealth-stealth] (30pt,90pt) -- (120pt,90pt);
	\node[anchor=south] at (75pt,97pt) {\scriptsize $n\!\cdot\!n^2$ edges between};
	\node[anchor=south] at (75pt,90pt) {\scriptsize $G$ and $C$};
	
	\draw[gray, thick] (95pt,12pt) -- (95pt,9pt) -- (145pt,9pt) -- (145pt,12pt);
	\draw[gray, thick] (120pt,6pt) -- (120pt,9pt);
	\node[anchor=north] at (120pt,7pt) {\scriptsize clique on};
	\node[anchor=north] at (120pt,0pt) {\scriptsize $n^2$ nodes};
	
	\node[anchor=south] at (14pt,77pt) {\small $G$};

\end{tikzpicture}
    }
	\caption{Illustration of the reduction to Maximum Independent Set: Given an original graph $G$ on $n$ nodes, we add a clique $C$ of size $n^2$ and connect every node of $G$ to every node of $C$.}
	\label{fig:addclique}
\end{figure}

\begin{lemma}
    \label{lem:1-2-np-hard}
    It is NP-hard to decide \twoclass{1,2}.
\end{lemma}
\begin{proof}
    We will show that \twoclass{1,2} is NP-hard by arguing that the reduction given above correctly decides the Maximum Independent Set problem on any given input graph $G$ in polynomial time, if we assume that \twoclass{1,2} can be decided in polynomial time. To begin with, note that the proposed reduction requires only polynomial time.
    
    As for the correctness, we assume the two colors are \emph{black} (proper color, i.e.\ an independent set) and \emph{white} (relaxed color) and argue that an optimal solution to the \twoclass{1,2} problem covering at least $c = k \cdot n^2$ edges exists if and only if there exists an independent set of size $k$ in $G$ (that may be colored black to establish the required conflict bound $c$).
    
    If there exists an independent set of size $k$ in $G$, then we can simply color this independent set in black and cover at least $k \cdot n^2$ edges from each of the $k$ vertices in $V_G$ to each of the $n^2$ vertices in $C$. Hence, there exists a solution to \twoclass{1,2} covering at least $k \cdot n^2$ edges.
    
    For the inverse direction, we first argue that the added clique $C$ can be assumed to be colored entirely white in an optimal solution. Assume for contradiction that it is strictly optimal to color at least one node in $C$ black. If the clique $C$ contains one black node, no other vertex of the graph $G'$ can be colored black as black is a proper color and each clique vertex is fully connected in the graph $G'$. Hence, such a coloring may cover at most $n^2 + n - 1$ edges. However, this can only be optimal for a maximum independent set of size $1$ in $G$; otherwise, we could simply color an independent set of size two in black and cover at least $2n^2 > n^2 + n - 1$ edges. For a maximum independent set of size $1$, $G$ must be a clique itself and thus, any node $v \in V_G$ covers $n^2 + n - 1$ edges as well. Hence, without loss of generality, we may assume that there exists an optimal solution of \twoclass{1,2} on $G'$ where the added clique $C$ is colored entirely white.
    
    It remains to show the following: If there exists a solution to the \twoclass{1,2} problem covering at least $k \cdot n^2$ edges in $G'$,
    there must exist an independent set of size $k$ in $G$. As argued above, there must exist at least one optimal solution $\mathcal{S}$ in which $C$ is colored entirely white. We show that $\mathcal{S}$ assigns at least $k$ vertices in $G$ the color black; in other words, showing that there exists an independent set of size $k$.
    
    Assume (for contradiction) that any maximum independent set on $G$ had size at most $k - 1$. Furthermore, note that $|E_G| < n^2$ as $G$ is a simple graph on $n$ vertices. Hence, the solution $\mathcal{S}$ for \twoclass{1,2} could cover at most $(k - 1) \cdot n^2 + |E_G| < k \cdot n^2$ edges -- a contradiction.
\end{proof}
Ultimately, we combine \Cref{lem:k-larger-2}, \Cref{lem:1-2-np-hard} and the known NP-completeness results of the \maxCut{} problem \cite{NPCompleteProblemsBook} to obtain:
\begin{theorem}
    For $k \geq 2$ and $(r, k) \neq (0, 2)$, it is NP-complete to decide whether an input graph $G$ admits a solution to \twoclass{r,k} covering at least $c$ edges.
\end{theorem}
\begin{remark}
    The reductions presented in \Cref{sec:NP-completeness} demonstrate that the decision problems \twoclass{r,k} are NP-complete for $k \geq 2$ and $(r, k) \neq (0, 2)$. Similarly, it can be shown that the naturally corresponding optimization problems are NP-hard. Thus, subsequently, we study the approximability of these optimization problems.
\end{remark}

\section{Approximability Results}
\label{sec:approximation}

In this section, we discuss the approximability of the optimal solution in \twoclass{r,k} problems. Our results show that the problems \twoclass{r,k} essentially behave like \maxKCut{} for $r \geq 2$, while they behave like Maximum Independent Set for $r = 1$. An illustration of the resulting approximation complexity classes can be found in \Cref{fig:approximation-classes}.

\begin{figure}[tbh]
    \centering
    \includegraphics[width=0.85\linewidth]{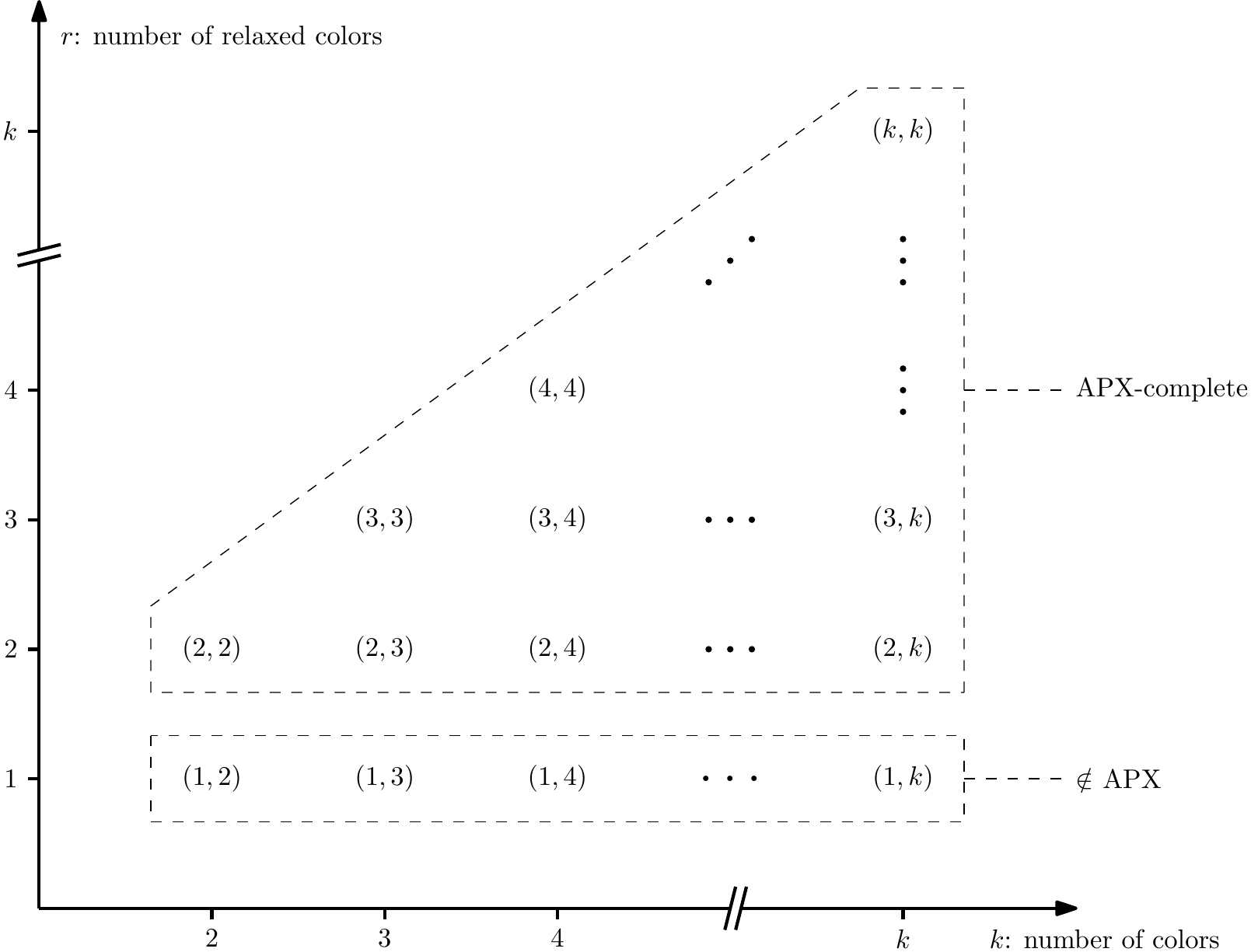}
    \caption{The \Twoclass{r,k} problems can be separated in two different approximation complexity classes. \Twoclass{r,k} is APX-complete for $r \geq 2$, while \Twoclass{1,k} cannot be approximated to any constant factor within polynomial time ($\notin$ APX).}
    \label{fig:approximation-classes}
\end{figure}

We assume some familiarity with the polynomial-time approximation complexity classes APX and PTAS. Intuitively, PTAS contains all problems that can be approximated to any constant factor, while APX contains all problems that can be approximated to some constant factor. Problems outside of APX can only be approximated to factors that depend on $n$.

We apply two kinds of reductions in this section. Both reductions reduce a (generally more well-known) problem $P_1$ to another problem $P_2$. Given a problem $P_1$ on an input graph $G$, our reduction first transforms $G$ into an instance of problem $P_2$ (in our case, another input graph $G'$). Then, assuming that an oracle provides us with a solution $\mathcal{S}$ of problem $P_2$ for the input graph $G'$, we show how to derive a solution $\mathcal{S}'$ for problem $P_1$ on graph $G$. Throughout this section, we leave it to the reader to verify that both transformations are computable in polynomial time.
\begin{definition}
    A reduction is called a \emph{continuous reduction} if there exists a constant $\alpha$ such that for any input graph $G$ and any possible solution $\mathcal{S}$ obtained on $G'$, we have
    \em 
        \[ \frac{\texttt{opt}_{P_1}(G)}{\texttt{val}_{P_1}(\mathcal{S}')} \quad \leq \quad \alpha \cdot \frac{\texttt{opt}_{P_2}(G')}{\texttt{val}_{P_2}(\mathcal{S})}. \]
    \em 
\end{definition}
Continuous reductions are known to preserve APX-membership; hence, if $P_1 \notin$ APX is known, then a continuous reduction from $P_1$ to $P_2$ proves that $P_2$ is not in APX either.

An \emph{L-reduction} (linear reduction), on the other hand, preserves membership in PTAS, and also preserves APX-completeness.
\begin{definition}
    We say that a reduction is an L-reduction if there are constants $\alpha$, $\beta$ such that the following two conditions hold:
    \begin{itemize}
        \item For any graph $G$, we have
            \em 
            \[ \texttt{opt}_{P_2}(G') \quad \leq \quad \alpha \cdot \texttt{opt}_{P_1}(G). \]
            \em 
        \item For any graph $G$ and possible solution $\mathcal{S}$ on $G'$, we have
            \em 
            \[ \texttt{opt}_{P_1}(G) - \texttt{val}_{P_1}(\mathcal{S}') \quad \leq \quad \beta \cdot \left( \texttt{opt}_{P_2}(G') - \texttt{val}_{P_2}(\mathcal{S}) \right). \]
            \em 
    \end{itemize}
\end{definition}
Note that these are simplified forms of the definitions, using the fact that we only study maximization problems in the paper. For a more detailed discussion of classes and tools in approximation complexity, we refer the reader to \cite{ApproxGuide}.

\subsection{Approximations for \twoclass{1,k}}
We show that the problems \twoclass{1,k} are intuitively more similar to the Maximum Independent Set (\maxIS{}) problem than to \maxCut{}. To that end, we first discuss \twoclass{1,2} separately and show that there exists a continuous reduction from \maxIS{} to the \twoclass{1,2} problem. Since it is known that \maxIS{} $\notin$ APX, this implies \twoclass{1,2} $\notin$ APX.

In fact, the reduction from \maxIS{} to \twoclass{1,2} follows in a rather straightforward way from a further analysis of the NP-completeness reduction presented in \Cref{sec:NP-completeness}. Recall that we transformed the input graph $G = (V_G, E_G)$ on $n$ nodes into a graph $G' = (V_{G'}, E_{G'})$ on $n^2 + n$ nodes by adding a clique $C$ of size $n^2$.
Given an algorithm to find an approximate solution $\mathcal{S}$ to problem \twoclass{1,2} in $G'$ in polynomial time, we derive an approximation for the \maxIS{} problem on the graph $G$ by selecting the nodes in $V_G$ that have been colored with the proper color. Once again, we assume the two colors are black (proper color) and white (relaxed color). The independent set of black nodes is denoted by $\mathcal{I}$.

We distinguish two cases. First, assume that the algorithm returns a coloring on $G'$ where $|\mathcal{I}|=1$. Recall that such a coloring can cover at most $n^2 + n - 1$ edges, as this is the highest possible degree in $G'$. For the \maxIS{} approximation, we then choose $\mathcal{I} = \{v\}$ some arbitrary $v \in V_G$.

Otherwise, at least two nodes are colored black. Recall that this implies that all black nodes are chosen from the original input graph $G$, as any node of the clique $C$ is connected to all other nodes in $G'$.
The independent set $\mathcal{I}$ covers at most $|\mathcal{I}| \cdot (n-|\mathcal{I}|)$ edges within $G$ (that is, of $E_G$), and exactly $|\mathcal{I}| \cdot n^2$ edges between $\mathcal{I}$ and $C$. Note that $|\mathcal{I}| \cdot (n-|\mathcal{I}|) \leq |E_G| < n^2$. In total, this amounts to at most $|\mathcal{I}|\cdot n^2 + n^2$ edges.

Since this bound equally holds for the case with $|\mathcal{I}| = 1$ above, we can establish
\[ \texttt{val}_{\twoclass{1,2}}(\mathcal{S}) \quad < \quad |\mathcal{I}|\cdot n^2 + n^2 \quad = \quad (|\mathcal{I}| + 1) \cdot n^2. \]

On the other hand, by choosing a maximum independent set in $G$ (and covering $n^2$ edges to the clique with each of its nodes), it follows that
\[ \texttt{opt}_{\twoclass{1,2}}(G') \quad \geq \quad \texttt{opt}_{\maxIS{}} (G) \cdot n^2. \]
Combining these two observations and interpreting the set of black nodes $\mathcal{I}$ as the approximate solution for the \maxIS{} problem on $G$, we get
\[ \frac{\texttt{opt}_{\twoclass{1,2}}(G')}{\texttt{val}_{\twoclass{1,2}}(\mathcal{S})} \geq \frac{ \texttt{opt}_{\maxIS{}} (G) \cdot n^2 }{(|\mathcal{I}| + 1) \cdot n^2} = \frac{\texttt{opt}_{\maxIS{}}(G) }{ |\mathcal{I}|+ 1} \overset{|\mathcal{I}|\ \geq\ 1}{\geq} \frac{ \texttt{opt}_{\maxIS{}}(G) }{ 2 \cdot |\mathcal{I}|} = \frac{1}{2} \cdot \frac{\texttt{opt}_{\maxIS{}}(G)}{\texttt{val}_{\maxIS{}}(\mathcal{I})}, \]
so the condition of continuous reduction is indeed satisfied for $\alpha=2$. Hence, the existence of a $\delta$-approximation for the \twoclass{1,2} problem would allow us to derive a $\delta / 2$-approximation for the \maxIS{} problem. Knowing that \maxIS{} $\notin$ APX, we conclude that \twoclass{1,2} $\notin$ APX.

\paragraph*{Generalization to \twoclass{1,k}}
We now present a continuous reduction from \twoclass{1,2} to \twoclass{1,k}. Together with our previous result, this shows that \twoclass{1,k} can also not be approximated to any constant factor.
\begin{lemma}
    There exists a continuous reduction from \twoclass{1,2} to  \twoclass{1,k} for any constant $k \geq 2$.
\end{lemma}
\begin{proof}
    Given a solution $\mathcal{S}$ of the \twoclass{1,k} problem on graph $G$, we show how to transform $\mathcal{S}$ into a solution $\mathcal{S'}$ for the \twoclass{1,2} problem on the same graph $G$, while only losing a constant approximation factor. Note that $k$ is assumed to be a constant.
    
    Generally, $\mathcal{S}$ can be regarded as a partitioning of the nodes into independent sets $I_1, ..., I_{k-1}$, and a relaxed color set $V_k$. Without loss of generality, we can assume that the sets $I_1, ..., I_{k-1}$ are sorted in descending order by the number of their outgoing edges $\text{out}(I_j)$, i.e., the number of edges that have exactly one endpoint in $I_j$.
    
    In order to turn this into a solution $\mathcal{S'}$ for the \twoclass{1,2} problem on the same graph $G$, let us select the independent set $I_1$ as the proper color class of $\mathcal{S'}$, and let the relaxed color class consist of all the remaining nodes, i.e.\ $V_0 := \left( \bigcup_{j=2}^{k-1} I_j \right) \cup V_k$. As $I_1$ is an independent set, the partitioning ($I_1$, $V_0$) is indeed a feasible solution of \twoclass{1,2}. Note that the value of this solution is by definition $\texttt{val}_{\twoclass{1,2}}(\mathcal{S'}) = \text{out}(I_1)$.
    
    In the original solution $\mathcal{S}$, since we only have one relaxed color, all covered edges have at least one of their endpoints in the sets $I_1, ..., I_{k-1}$, and thus
    \[ \texttt{val}_{\twoclass{1,k}}(\mathcal{S}) \quad \leq \quad \sum_{j=1}^{k-1} \text{out}(I_j) \quad \leq \quad (k-1) \cdot \text{out}(I_1), \]
    because we assumed that $I_1$ is the independent set with the maximum number of outgoing edges $\text{out}(I_1)$. This shows that
    \[ \frac{1}{k-1} \cdot \texttt{val}_{\twoclass{1,k}}(\mathcal{S}) \quad \leq \quad \texttt{val}_{\twoclass{1,2}}(\mathcal{S'}). \]
    Since the set of feasible colorings in \twoclass{1,2} is a subset of those in \twoclass{1,k}
    \[ \texttt{opt}_{\twoclass{1,2}}(G) \quad \leq \quad \texttt{opt}_{\twoclass{1,k}}(G) \]    follows straightforwardly. Thus we have
    \[ \frac{\texttt{opt}_{\twoclass{1,2}}(G)}{\texttt{val}_{\twoclass{1,2}}(\mathcal{S'})} \quad \leq \quad \frac{\texttt{opt}_{\twoclass{1,k}}(G)}{\frac{1}{k-1} \cdot \texttt{val}_{\twoclass{1,k}}(\mathcal{S})} \quad = \quad (k-1) \cdot \frac{\texttt{opt}_{\twoclass{1,k}}(G)}{\texttt{val}_{\twoclass{1,k}}(\mathcal{S})}, \]
    showing that this is indeed a continuous reduction for $\alpha=(k-1)$.
\end{proof}
Since we have already seen that \twoclass{1,2} $\notin$ APX, this implies the following theorem.
\begin{theorem}
    For any constant $k \geq 2$, the problem \twoclass{1,k} $\notin$ APX.
\end{theorem}

\subsection{Approximations for \twoclass{r,k} with $r \geq 2$}
For this case, recall that \twoclass{k,k} is essentially a reformulation of the \maxKCut{} problem. Since \maxKCut{} is known to be APX-complete \cite{APXcomplete}, the same naturally holds for \twoclass{k,k}. Let us consider some other $r \in \{ 2, ..., k-1\}$.

We begin by showing that \twoclass{r,k} is also contained in APX.
It is known that using the probabilistic method and derandomization, one can obtain a simple deterministic algorithm that finds a partitioning ($V_1$, $V_2$) which cuts at least $\frac{|E|}{2}$ edges \cite{ProbMethod}. This already ensures that the algorithm in question returns a $\frac{1}{2}$-approximation for the \maxCut{} problem.

Furthermore, this technique also allows us to approximate the optimum solution of each \twoclass{r,k} (with $r \geq 2$) to a factor of $\frac{1}{2}$. Given a problem \twoclass{r,k} on $G$, we can simply run this method to obtain a partitioning ($V_1$, $V_2$) with at least $\frac{|E|}{2}$ edges cut, and assign the first two relaxed colors to these partitions (not using the remaining $k-2$ colors at all). As we have $\texttt{opt}_{\twoclass{r,k}}(G) \leq |E|$ in any case, this is already a $\frac{1}{2}$-approximation algorithm for the problem, showing that \twoclass{r,k} can always be approximated to a constant factor.
\begin{lemma} \label{lem:inAPX}
For $k \geq r \geq 2$, we have \twoclass{r,k} $\in$ APX.
\end{lemma}
In order to show that \twoclass{r,k} is also complete in APX, we present an L-reduction to transform a solution of \twoclass{r,k} into a solution of \twoclass{2,2}, i.e.\ \maxCut{}. The APX-completeness of \maxKCut{} has also been shown through a reduction to \maxCut{} in \cite{APXcomplete}. We essentially show that the same graph transformation can be applied in our case as in the reduction of \cite{APXcomplete}, regardless of the difference between \maxKCut{} and \twoclass{r,k}, i.e.\ the extra restriction that some of the classes must be proper colors. For this, we first describe the reduction of \cite{APXcomplete}.
\begin{lemma}
    There exists an L-reduction from \twoclass{2,2} to \twoclass{k,k} (based on the transformation idea described in \cite{APXcomplete}).
\end{lemma}
\begin{proof}
    We describe the reduction in two steps. We first present a reduction to the \twoclass{r,k} problem on multigraphs (i.e.\ graphs allowing parallel edges). We then discuss how to transform such a multigraph into a simple graph without parallel edges.
    
    Given a \maxCut{} problem on an input graph $G$ with $n$ nodes and $m$ edges, let us create a multigraph $G_M'$ as follows. We add a clique $C$ of $k-2$ further nodes to $G$, with each edge between the nodes of $C$ having a multiplicity of $2m$. Furthermore, we connect each node $v \in G$ to each node $u \in C$, with an edge of multiplicity $2d(v)$. The resulting multigraph consist of $n+k-2$ nodes and $|E_{G_M'}| = m + (k - 2) \cdot 4m + {k-2 \choose 2} \cdot 2m$ edges (with multiplicity).
    
    Note that this already satisfies the first condition of the L-reduction. As we discussed earlier, an optimum solution of \maxCut{} on $G$ will certainly cover at least $\frac{m}{2}$ edges. On the other hand, $G_M'$ has only $|E_{G_M'}| \in \mathcal{O}(k^2m)$ edges. Thus, we have
   \[ \frac{1}{2} \cdot \texttt{opt}_{\twoclass{r,k}}(G_M') \quad \leq \quad \mathcal{O}(k^2) \cdot \texttt{opt}_{\maxCut{}}(G)\,, \]
    which implies the first condition of the L-reduction for some $\alpha = 2 \cdot \mathcal{O}(k^2)$.

    We now show that the second condition also holds for $\beta=1$. Note that if we color the nodes of $G$ with two relaxed colors in an arbitrary way, and use each of the remaining $k-2$ (proper or relaxed) colors to color one node of $C$, then we immediately obtain a coloring with at
    most $m$ uncovered edges.
    
    Consider a solution $\mathcal{S}$ to the \twoclass{r,k} problem, that is, a $k$-coloring of the nodes of $G_M'$. If such a coloring was to assign the same color to any two nodes of $C$, then it would lose at least $2m$ edges, and thus the right side of the second condition would be at least $m$, implying that the condition holds regardless of our choice of $\mathcal{S}'$. Thus, it remains to show the condition for the case when $\mathcal{S}$ colors the nodes of $C$ with $k-2$ different colors.
    
    Note that if any $v \in G$ is assigned one of the $k - 2$ colors used in $C$, it may be recolored to one of the remaining 2 colors (not used in $C$). Thereby, the value of a solution $\mathcal{S}$ will increase by at least $2 d(v)$, while introducing at most $d(v)$ new conflicts. Hence, after recoloring all such vertices in $G$, we can assume to obtain a solution $\mathcal{S}$ where $C$ is colored with $k-2$ different colors, and only the remaining 2 colors are used in $G$.
    
    
    Given a $k$-coloring $\mathcal{S}$ of $G_M'$, we can then simply define the solution $\mathcal{S}'$ of \maxCut{} as the restriction of $\mathcal{S}$ to the nodes of $G$. Since both $\mathcal{S}$ and the optimal coloring covers all edges of $G_M'$ outside of $G$, and they both use only 2 colors in $G$, their difference is in fact determined by how many edges they cover in $G$. Hence, we have
    \[ \texttt{opt}_{\twoclass{r,k}}(G_M') - \texttt{val}_{\twoclass{r,k}}(\mathcal{S}) \quad = \quad \texttt{opt}_{\maxCut{}}(G) - \texttt{val}_{\maxCut{}}(\mathcal{S}'), \]
    and thus the second condition indeed holds with $\beta=1$.
    
    As a second step, we transform the multigraph $G_M'$ into a simple graph $G'$ with the same properties. For this, we define an auxiliary graph $K$ on $k+2$ nodes. Let us denote the nodes of $K$ by $w_1$, $w_2$, ... $w_{k+2}$, and let $K$ be the graph obtained by taking a clique on these $k+2$ nodes, and then deleting the edges ($w_1,w_2$), ($w_2,w_3$) and ($w_3,w_4$) from this clique. There are two crucial properties of the resulting graph $K$: (i) it is $k$-colorable without any conflicts, (ii) every valid $k$-coloring of $K$ assigns different colors to nodes $w_2$ and $w_3$.
    
    We then use these copies of the graph $K$ to replace the parallel edges of $G_M'$. For each edge $(u,v)$ in $G_M'$ which is not contained in $G$, we can replace $(u,v)$ by a separate instance of $K$, with the nodes $u$ and $v$ taking the roles of nodes $w_2$ and $w_3$. Since we only insert $k$ new nodes and $\mathcal{O}(k^2)$ new edges for each original edge $(u,v)$ of $G_M'$, the number of nodes and edges in the resulting graph $G'$ will still remain polynomial in $n$. In the graph obtained, each $K$ will behave like an edge in the following sense: If nodes $u$ and $v$ have different color, then it is possible to color $K$ without a conflict; but if $u$ and $v$ have the same color, then there will be at least one conflict in $K$. Thus, if there were originally $t$ parallel edges between $u$ and $v$, then assigning the same color to $u$ and $v$ implies that there will certainly be at least $t$ conflicts within the $t$ copies of $K$ between the two nodes.
    
    For the purpose of the reduction, the obtained graph $G'$ behaves exactly like the multigraph~$G_M'$. Since the number of edges in $G'$ is in $\mathcal{O}(k^4m)$, the first condition still holds. In any solution $\mathcal{S}$, if a copy of $K$ has more than one conflict, we can always recolor the nodes $w_1$, $w_4$, $w_5$, ..., $w_{k+1}$ such that there are only 0 or 1 conflicts in the end, depending on whether $w_2$ and $w_3$ have the same color or not. Hence, we can assume that each $K$ represents exactly as many conflicts as the parallel edge it has replaced. Thus, the reduction described for the multigraph case also works for $G'$ without modification.
\end{proof}
The transformation described so far is a valid reduction to \twoclass{k,k}, i.e.\ when all the available colors are relaxed. We now show that it also remains a valid reduction for \twoclass{r,k}.
\begin{lemma}
    There exists an L-reduction from \twoclass{2,2} to \twoclass{r,k} for any $r \geq 2$.
\end{lemma}
\begin{proof}
    For the first part of the reduction (to the multigraph $G_M'$), we can argue as before that all nodes of $C$ obtain a different color in $\mathcal{S}$. We then use the same technique to ensure that the nodes of $G$ will only get the remaining two colors. This technique still increases the number of covered edges in each step, but might lead to an invalid coloring in the end if one (or both) of the two colors used in $G$ are among the proper colors. However, if one of the proper colors used in $G$ has a conflict, then we can simply swap this proper color with a relaxed color that is currently used only on a single node in $C$. After at most 2 such swaps, we can make sure that both colors used to color $G$ are relaxed colors. We can consider the resulting coloring as the original solution $\mathcal{S}$, and apply the same reduction as before.
    
    As for the transformation to a simple graph, whenever two neighboring nodes $u$ and $v$ have different colors, it is always possible to color the whole copy of graph $K$ between them without any conflicts at all, so it does not matter whether the colors used on the specific nodes of $K$ are proper or relaxed colors. If $u$ and $v$ have the same color $c_1$, then there is at least one relaxed color $c_2 \neq c_1$. We can then color the copy of $K$ by using $c_2$ on nodes $w_1$ and $w_4$, and each of the remaining $k-2$ colors on one of the remaining $k-2$ nodes; in this coloring, the only conflict will be on the edge $(w_1, w_4)$. Thus it is still always possible to color each copy of $K$ with exactly 0 or 1 conflicts (depending on whether $u$ and $v$ have the same color), regardless of the colors of $u$ and $v$ being proper or relaxed.
\end{proof}
Altogether, this shows the following.
\begin{theorem}
    For any $k \geq r \geq 2$, the problem \twoclass{r,k} is APX-complete.
\end{theorem}

\section{Conclusion -- and a Generalization}
In the paper, we analyzed the problem of \Twoclass{r,k}, which aims to maximize the number of covered edges when coloring a graph with $r$ relaxed and $k-r$ proper colors. We have seen that \twoclass{1,2} is similar but not identical to the Maximum Independent Set problem. We have shown that the problem \twoclass{r,k} is NP-complete for any $k \geq 2$ (except when $(r,k) = (0,2)$). Furthermore, we also proved that the problems \twoclass{r,k} with $r \geq 2$ are APX-complete, and that \twoclass{1,k} cannot be approximated to any constant factor in polynomial time.

Throughout this paper, we have established the close relation between the \twoclass{r,k} and the \maxKCut{} problem. Note that the cost function used for the optimization is the same for both problems, while they only differ in their set of feasible solutions.

In practice, however, one may think of other reasonable cost functions. For instance, Defective Coloring exclusively focuses on the node with the highest number of conflicts, and does not evaluate the coloring on the rest of the graph.

Considering application areas, e.g.\ modeling interference in frequency allocation problems, both approaches might be reasonable. A provider would generally want to keep the total number of conflicts in the graph small, as this corresponds to the overall service quality. On the other hand, a provider would probably also try to minimize the maximal number of conflicts at a single node in order to avoid upsetting any particular costumer.

In fact, an interesting generalization allows us to model both settings (and the trade-off between them) at the same time. Given a graph $G = (V,E)$ and a $k$-coloring of $G$, let us denote the number of conflicts at node $v \in V$ (i.e., the number of conflict edges incident to~$v$) by $\kappa(v)$. Then, given a real parameter $p > 0$, we define the \emph{Generalized Conflict Coloring} problem, where the cost of a coloring $\mathcal{S}$ is defined as
\[ \texttt{cost}(\mathcal{S}) \quad := \quad \sum_{v \in V} \kappa(v)^p . \]
The goal of the Generalized Conflict Coloring problem is to find a coloring $\mathcal{S}$ on $k$ colors that minimizes this cost.
This generalized cost function can be combined with any restrictions on the set of feasible colorings, for example, the family of \twoclass{r,k} problems.

In the following, we motivate that this problem is a natural generalization of various known coloring problems. In the case when $p=1$, the aim of the problem is to minimize $\sum_{v \in V} \kappa(v)$, which is simply two times the number of monochromatic edges (conflicts) in the graph. Hence the cost of a coloring $\mathcal{S}$ is exactly $2 \cdot \left(|E| - \texttt{val}_{\twoclass{k,k}}(\mathcal{S}) \right)$, and thus the problem is identical to \Twoclass{k,k} (i.e., to \maxKCut{}).

On the other hand, as $p$ goes to infinity, the sum in the cost will be dominated only by the highest $\kappa(v)$ value in the graph, and thus a coloring $\mathcal{S}_1$ will be better than a coloring $\mathcal{S}_2$ exactly if the maximal $\kappa(v)$ is smaller in $\mathcal{S}_1$. This is precisely what is minimized in Defective Coloring; hence, Defective Coloring is also obtained as a subcase of Generalized Conflict Coloring in the limit for $p \to \infty$.

Finally, if $p \to 0$, the difference between the distinct $\kappa(v)$ values will diminish for all $\kappa(v) > 0$, since all are approaching 1 for $p$ small enough. Thus, in the limit we obtain
\[\lim_{p \to 0} \texttt{cost}(\mathcal{S}) \quad = \quad \sum_{\substack{v \in V \\ \kappa(v)\,>\,0}} 1 . \]
In other words, the cost of a coloring $\mathcal{S}$ becomes the number of nodes for which $\kappa(v) > 0$. Minimizing this corresponds to finding a coloring which leaves the highest possible number of nodes without a conflict.

Furthermore, intermediate values of $p$ allow us to study the trade-off between these problems. For instance, if $p=2$ is selected, then we obtain a setting between \maxKCut{} and Defective Coloring, where we have to optimize both for a small number of total conflicts and a reasonably balanced distribution of these conflicts among nodes. Altogether, we conclude that the further study of this Generalized Conflict Coloring problem is not only interesting from a theoretical point of view, but also motivated by practical applications.


\bibliography{main}

\end{document}